\documentclass[twoside,leqno,11pt]{article}
\usepackage{graphicx,fullpage} 
\usepackage{amsfonts,amsmath,amsthm,amssymb,mathtools, comment}\usepackage{xcolor}
\usepackage[T1]{fontenc}
\usepackage[colorlinks,citecolor=blue,linkcolor=blue,urlcolor=red,pagebackref]{hyperref}
\usepackage[capitalise]{cleveref}
\usepackage{url}
\usepackage{thmtools, thm-restate}
\usepackage{algorithm, algorithmicx, algpseudocode}
\usepackage{tikz}
\usepackage{bm}
\usepackage{thm-restate}
\usepackage{mdframed}

\newcommand{\eps}{\varepsilon}
\renewcommand{\Pr}{\operatorname*{\mathbf{Pr}}}

\newcommand{\poly}{\operatorname{\mathrm{poly}}}
\newcommand{\polylog}{\poly\log}

\newcommand{\E}{\mathbb{E}}

\newcommand\threesum{$\mathsf{3SUM}$}
\newcommand\threexor{$\mathsf{3XOR}$}

\newtheorem{theorem}{Theorem}[section]  
\newtheorem{lemma}[theorem]{Lemma}

\newtheorem{claim}[theorem]{Claim}
\newtheorem{cor}[theorem]{Corollary}
\theoremstyle{definition}
\newtheorem{definition}[theorem]{Definition}

\newif\ifcomments
 \commentsfalse  

\title{Preprocessed 3SUM for Unknown Universes\\ with Subquadratic Space}
\author{Yael Kirkpatrick\\MIT \and John Kuszmaul\\MIT \and Surya Mathialagan\\NTT Research \and Virginia Vassilevska Williams\\MIT}
\date{}

\begin{document}
    \maketitle

\begin{abstract}
    We consider the classic \threesum{} problem: given sets of integers $A, B, C $, determine whether there is a tuple $(a, b, c) \in A \times B \times C$ satisfying $a + b + c = 0$. The 3SUM Hypothesis, central in fine-grained complexity, states that there does not exist a truly subquadratic time 3SUM algorithm. Given this long-standing barrier, recent work over the past decade has explored \threesum{} from a data structural perspective. Specifically, in the \threesum{} in preprocessed universes regime, we are tasked with preprocessing sets $A, B$ of size $n$, to create a space-efficient data structure that can quickly answer queries, each of which is a \threesum{} problem of the form $A', B', C'$, where $A' \subseteq A$ and $B' \subseteq B$. A series of results have achieved $\tilde{O}(n^2)$ preprocessing time, $\tilde{O}(n^2)$ space, and query time improving progressively from $\tilde{O}(n^{1.9})$ \cite{ChanL15} to $\tilde{O}(n^{11/6})$ \cite{ChanWX23} to $\tilde{O}(n^{1.5})$ \cite{preprocessedSOSA2025}.
    Given these series of works improving query time, a natural open question has emerged: can one achieve both
    truly subquadratic space and truly subquadratic query time for \threesum{} in preprocessed universes?

    We resolve this question affirmatively, presenting a tradeoff curve between query and space complexity. Specifically, we present a simple randomized algorithm achieving $\tilde{O}(n^{1.5 + \varepsilon})$ query time and $\tilde{O}(n^{2 - 2\varepsilon/3})$ space complexity. Furthermore, our algorithm has $\tilde{O}(n^2)$ preprocessing time, matching past work. Notably, quadratic preprocessing is likely necessary for our tradeoff as either the preprocessing or the query time must be at least $n^{2-o(1)}$ under the \threesum{} Hypothesis.
\end{abstract}
    \section{Introduction}
One of the central problems in fine-grained complexity is \threesum{}: Given three sets of $n$ integers, $A,B,C$, are there $a\in A, b\in B, c\in C$ so that $a+b+c=0$? A classic algorithm solves \threesum{} in $O(n^2)$ time. The fastest known algorithm by Baran, Demaine, Patrascu \cite{bdp} runs in $n^2\poly(\log\log n) / \log^2 n$ time in the word-RAM model. Chan \cite{chanreal} obtained the same running time in the real-RAM model. 
The \threesum{} Hypothesis from fine-grained complexity postulates that there is no $O(n^{2-\eps})$ time algorithm for \threesum{} for any $\eps>0$. 
Given this algorithmic barrier, much research has gone into understanding settings in which \threesum{} admits faster algorithms (see e.g., \cite{ChanL15,JinX23,GoldsteinKLP17,kp19,dataStructuresFor3sum2020, ChanWX23, preprocessedSOSA2025}).

In this paper we consider the \threesum{} problem in {\em preprocessed universes}, a variant of which was first discussed by Bansal and Williams \cite{bw12} and attributed to Avrim Blum.

The problem has two variants:
\begin{itemize}
\item {\bf \threesum{} in {\em preprocessed universes} with known} $C$:
Given three sets of $n$ integers $A,B,C$, preprocess them, so that on any query $A'\subseteq A, B'\subseteq B, C'\subseteq C$, one can solve \threesum{} on $A',B',C'$ fast. 

\item {\bf \threesum{} in {\em preprocessed universes} with unknown} $C$:
Given two sets of $n$ integers $A,B$, preprocess them, so that on any query $A'\subseteq A, B'\subseteq B$ and set of $n$ integers $C'$, one can solve \threesum{} on $A',B',C'$ fast.
\end{itemize}
For both variants there are three natural parameters to optimize: (1) the preprocessing time, (2) the space used after preprocessing, (3) the query time.

Note that under the \threesum{} Hypothesis, either the preprocessing time or the query time must be $n^{2-o(1)}$. Hence all prior work has focused on optimizing the query time and space usage, given that the preprocessing time is essentially quadratic. A summary of the prior work is given in Table~\ref{table:results}.

\begin{table}[h!]
\centering
\begin{tabular}{||c | c | c | c||} 
 \hline
 Reference & Query & Space & Unknown $C$ \\ [0.5ex] 
 \hline\hline
 \cite{ChanL15} & $n^{13/7}$ & $n^{13/7}$  &  NO \\ 
  \hline
\cite{ChanWX23} & $n^{1.891}$ & $n^{1.891}$  &  NO \\ 
 \hline
 \cite{preprocessedSOSA2025} & $n^{1.5}$ & $n^{1.5}$ & NO \\
 \hline
 \cite{ChanL15} & $n^{1.9}$ & $n^2$ & YES  \\

 \hline
 \cite{ChanWX23} & $n^{11/6}$ & $n^2$ & YES  \\

 \hline
 \cite{preprocessedSOSA2025} & $n^{1.5}$ & $n^{2}$ & YES \\
 \hline
 {\bf This work} & & & \\
  $\forall \eps\in [0,{1/2}]$  & $n^{1.5+\eps}$ & $n^{2-2\eps/3}$ & YES  \\
 \hline
\end{tabular}
\caption{All running times above include polylogarithmic factors, so they should be thought of as enclosed in $\widetilde{O}()$. }
\label{table:results}
\end{table}


For the {\em known} $C$ case, Chan and Lewenstein \cite{ChanL15} showed that after $\widetilde{O}(n^2)$ preprocessing time\footnote{$\widetilde{O}$ suppresses $\polylog(n)$ factors.}, the \threesum{} queries can be performed in truly subquadratic time of $\widetilde{O}(n^{13/7})$, while the space usage is also subquadratic, namely $\widetilde{O}(n^{13/7})$. Fischer \cite{Fischer25} extended this to the real RAM model. Chan, Vassilevska Williams and Xu \cite{ChanWX23} 
improved the query time and space usage to $O(n^{1.891})$. With quadratic space usage they also obtained improved query time of $\widetilde{O}(n^{11/6})$, and this algorithm also worked for the {\em unknown} $C$ case for which \cite{ChanL15} had achieved $\widetilde{O}(n^{1.9})$ query time and quadratic space.

More recently, with a much simpler algorithm, Kasliwal, Polak and Sharma \cite{preprocessedSOSA2025} drastically improved the query running time and space usage to $O(n^{1.5}\log n)$ for the {\em known} $C$ case, while keeping the preprocessing time quadratic. For the {\em unknown} $C$ case, their approach achieves the same query time but uses quadratic space.

Notably, all known quadratic preprocessing and subquadratic query time algorithms for \threesum{} in preprocessed universes in the case of {\em unknown} $C$ use {\bf quadratic space}. This naturally raises the following question.

\begin{center}
{\em Can one achieve essentially quadratic preprocessing time and {\bf truly subquadratic space and query time} for \threesum{} in preprocessed universes with unknown $C$?}
\end{center}

More generally, we are interested in the trade-off between space usage and query time for the problem.
Our main result is as follows:
\begin{restatable}{theorem}{SubquadraticUnknownC}
\label{thm:subquadraticUnknownC}
   For any $\eps \in [0,1/2]$, there exists a randomized algorithm that preprocesses $2$ sets of $n$ integers $A,B$ in $\widetilde{O}(n^2)$ time using $\widetilde{O}(n^{2-2\eps/3})$ space, such that upon a query $A'\subseteq A, B'\subseteq B$ and a set of $O (n)$ integers $C'$, it solves \threesum{} on $A',B',C'$ in $\widetilde{O}(n^{1.5+\eps})$ time with high probability.
\end{restatable}

We note that the query time of our algorithm is only guaranteed against an oblivious adversary. While not stated explicitly, this is the case for past works on preprocessed \threesum{} as well.

\subsection{Technical Overview}
In this paper we combine the techniques of two previously best known results in the field of preprocessed \threesum{} to achieve the first algorithm for preprocessed \threesum{} with unknown $C$ using subquadratic space. 

\paragraph{Prior work on Preprocessed \threesum{}.} First we consider the recent result of Kasliwal, Polak and Sharma \cite{preprocessedSOSA2025}, where in the case where $C$ is known at preprocessing, the authors obtain an algorithm with query time $\widetilde{O}(n^{1.5})$ and $\widetilde{O}(n^{1.5})$ space. This algorithm is achieved by sampling a random prime $p$ and computing the set of triples $a,b,c\in A\times B\times C$ that are not a \threesum{} solution, $a+b\neq c$, but mod p would look like one, i.e., $a+b \equiv_p c$. Such a triple is considered a \emph{false positive} and the number of false positives is bounded by $O(\frac{n^3}{p})$. 

Given the subsets $A',B',C'$ at query time, the algorithm computes $A'+B'$ mod $p$ in $O(p\log p)$ time using the fast Fourier transform (FFT) - for every $c\in C$ it computes the number of pairs $(a,b)\in A'\times B'$ such that $a+b\equiv_p c$. It then goes over all stored false positives and for each false positive in $A'\times B'\times C'$, subtracts one from the computed number of pairs in $A'\times B'$ that sum to the specified value of $c$ mod $p$. This leaves us with the number of pairs that truly sum up to $c$. The algorithm runs in time $\widetilde{O}(p + \frac{n^3}{p})$, and so taking $p=O(n^{1.5})$ gives the desired query time and space.

In the case where $C$ is unknown at preprocessing, the authors use a similar approach of FFT and ruling out false positives. However, in this case, the algorithm stores a data structure that essentially saves the entire sumset $A+B$, which we will have to avoid in order to achieve subquadratic space.

\paragraph{Prior work on \threesum{}-indexing.} The second result we consider is the work of Golovnev, Guo, Horel, Park and Vaikuntanathan \cite{dataStructuresFor3sum2020}, where the authors construct an algorithm for the closely related problem of \threesum{}-indexing. In this variant of the problem, we are given the sets $A$ and $B$ to preprocess, and at query time we receive a value $c$ and are asked if there exists a pair $a\in A,b\in B$ that sums to it. The question in this case regards the optimal tradeoff between query time and space. The key difference between \threesum{} indexing and preprocessed \threesum{} with unknown $C$ is that for the former, the query is about any pair of numbers in $A\times B$, whereas in the latter, the query asks to find a \threesum{} solution within a specific subset of $A$ and $B$. Thus, \threesum{}-indexing cannot directly solve preprocessed \threesum{} with unknown $C$ and some additional tools are needed in order to apply \threesum{} indexing techniques to the problem of preprocessed \threesum{}. 

To solve the \threesum{} indexing problem, Golovnev et al. \cite{dataStructuresFor3sum2020} use a hashing argument combined with the cryptographic tool of Fiat-Naor function inversion \cite{fiatNaor}. This well known result gives a query-time/space tradeoff for the data structures problem of inverting a function. Golovnev et al. consider the function $f:A\times B \to A+B$ defined by $f(a,b) = a+b$. Inverting this function is equivalent to finding a pair in $A\times B$ that sums to a given value. However, Fiat-Naor function inversion on its own is insufficient for our problem, as it does not allow for restriction of the sources when multiple candidates exist, which we would need in the case when multiple pairs sum up to a given value but some of them are not in $A'\times B'$. Thus, this technique cannot be directly applied to the more general problem of preprocessed \threesum{}.

\paragraph{Our techniques.} To combine these techniques, we treat the values in the set $A+B$ differently based on their multiplicity, i.e., the number of pairs in $A,B$ that sum to each value. We call a value $c$ a \emph{heavy hitter} if it has multiplicity at least $n^\delta$ and denote by $\widetilde{C}$ the set of all heavy hitters. Note that since every value in $\widetilde{C}$ has more than $n^\delta$ different pairs that sum to it, we can bound the size of the set by $|\widetilde{C}|= O(n^{2-\delta})$. To handle the heavy hitters, we apply the techniques of \cite{preprocessedSOSA2025} and precompute the false positives mod $p$ in $A\times B\times \widetilde{C}$. Using a standard argument we can show that the number of such false positives can be bounded by $O(\frac{n^{4-\delta}}{p})$.

Next, we handle the non-heavy hitters. We divide our input set $A$ into $\ell =\widetilde{O}(n^\delta)$ parts $A_1, \ldots, A_\ell$. By doing this multiple times, we can guarantee that with high probability, for every value $c\in A+B$ with multiplicity $<n^\delta$, in at least one of the instances there are no two pairs $(a_1,b_1),(a_2,b_2)\in A\times B$ that sum to $c$ and fall into the same part in the partition, i.e., no $i$ satisfies $a_1, a_2\in A_i$. We now preprocess a Fiat-Naor inversion on each $A_i\times B$ independently, adjusting the function $f(a,b)=a+b$ slightly so that we can use a version of Fiat-Naor with a better time/space tradeoff than the general case used in \cite{dataStructuresFor3sum2020}.

Given the sets $A',B',C$ at query time, we compute $A'+B'$ mod $p$ using FFT. Now, for any heavy hitter $c\in C\cap \widetilde{C}$, we rule out the false positives in $A'\times B'$ that sum to $c$ mod $p$ and are left with the number of pairs that sum to $c$ exactly. For every non-heavy hitter $c\in C\setminus \widetilde{C}$, we query each Fiat-Naor instance to find a pair $(a,b)\in A_i\times B$ that sums to $c$ if such a pair exists. We then check if $(a,b)\in A'\times B'$ and if so return the triple $a,b,c$. Since each pair that sums to $c$ falls in a distinct Fiat-Naor instance, this process will scan all such pairs and find one in $A'\times B'$ if one exists.

The space requirements of our algorithm are dominated by the $\widetilde{O}(n^\delta)$ instances of Fiat-Naor function inversion data structures and the set of $O(\frac{n^{4-\delta}}{p})$ false positives for heavy hitters. By balancing the values of $\delta,p$ we obtain a tradeoff between query time and space, achieving subquadratic space for any query time greater than $n^{1.5+\Omega(1)}$.
    \section{Preliminaries}
Let $[n]\coloneqq\{1,2,\dots,n\}$. Denote by $\log$ the base-2 logarithm and by $\ln$ the natural log.

Given two sets of numbers $A,B$ we define their sumset as $A+B = \{a+b : a\in A, b\in B\}$. We define the {\em multiplicity} of an element in $A+B$ as the number of pairs in $A\times B$ that sum up to it $\text{mult}(c)\coloneqq|\{(a,b)\in A\times B : a+b = c\}|$. We denote by $\equiv_p$ equivalence of integers modulo $p$. We assume the elements of $A, B, C$ (and the sumset $A + B$) are bounded by $n^{\alpha} = \poly(n)$ for some $\alpha = O(1)$, and thus the elements fit into a constant number of machine words. By shifting the values by a constant we can also assume they are all positive.


\subsection{Inequalities} 
We first recall the following inequalities.

\begin{lemma}[Markov's Inequality]
    Let $X$ be a non-negative random variable and $a\geq 0$, then
    \[
    \Pr[X\geq a\cdot \mathbb{E}[X]] \leq \frac{1}{a}.
    \]
\end{lemma}

\begin{theorem}[Chernoff bound]
    Let $X = \sum_{i=1}^n X_i$, where each $X_i$ is an independent Bernoulli variable where $X_i = 1$ with probability $p_i$. Let $\mu = \mathbb{E}[X]$. Then, for all $\delta > 0$,
    \[
        \Pr[X \geq (1 + \delta)\mu] \leq e^{-\frac{\delta^2\mu}{\delta + 2}}.
    \]
\end{theorem}

\begin{lemma}[Bernoulli's Inequality]
    For $y \in [0, 1)$ and $x \geq 1$, we have that
    \[
        (1 - y)^x \geq 1 - xy.
    \]
\end{lemma}

\subsection{Preprocessed \threesum{} Techniques}

We present two key tools used in preprocessed \threesum{} below as lemmas. Attributed to folklore, these lemmas were presented explicitly in \cite{preprocessedSOSA2025}.

\paragraph{\boldmath Computing sumsets mod $p$ via FFT.}
\begin{lemma}[{\cite[Lemma $2.1$]{preprocessedSOSA2025}}]
\label{lm:fft}
    Given sets $A, B$ of at most $n$ integers, and a positive integer $p$, we can compute the multiset $H \coloneq A + B \pmod p$ in time $O(n + p\log p)$.
\end{lemma}

We briefly note that computing the multiset of $A + B$ mod $p$ can be reduced to polynomial multiplication of two degree $p$ polynomials, which can be quickly computed via FFT (see \cite{preprocessedSOSA2025} for a detailed proof). Given the ease of sumset computation for a small modulus, a natural obstacle that arises in preprocessed \threesum{} is the occurrence of false positives, i.e., triplets of the form $(a, b, c) \in A \times B \times C$ such that $a + b \equiv_p c $ but $a + b \neq c$. The following lemma bounds the number of such false positives.


\paragraph{\boldmath Bounding false positives mod $p$.}
\begin{lemma}[{\cite[Lemma 2.2]{preprocessedSOSA2025}}] \label{lm:falsepositives}
    Let $A, B, C \in [U]$ be sets of integers bounded by $U$ such that $\forall (a, b, c) \in A \times B \times C$, we have $a + b \neq c$. Let $p$ be a prime sampled uniformly from $[R, 2R]$, for some positive integer $R$. The expected number of false positives is
    \[\E[\#\{(a, b, c) \in A \times B \times C : a + b \equiv_p c \}] \leq O\left(\frac{|A||B||C| \log U}{R}\right).\]

\end{lemma}

The proof (see \cite{preprocessedSOSA2025}) proceeds by a counting argument, with the key observation that only $\log_r(2U)$ of the primes in the range $[r, 2r]$ can divide $|a + b - c|$. We observe by Markov's inequality that at least $1/2$ of the primes in this range must have at most $O(|A||B||C|\log U / r)$ false positives.




We note that the set of false positives $F \coloneqq \{(a,b,c)\in A\times B \times C : a+b\equiv_p c \land a+b\neq c\}$ can be constructed in $O(|A||B| + |C| + r + |F|)$ time  (see the preprocessing step under Theorem $3$ of \cite{preprocessedSOSA2025}). This will be $O(n^2 + |F|)$ time for the case that $|A||B|, |C|, r = O(n^2)$, as we have in Section \ref{sec:main-result}. See \cite[Lemma 3.2]{FischerKaPo24} for a deterministic variant of \cref{lm:falsepositives}.


\subsection{\threesum{} Indexing Techniques}
The main technique used for \threesum{} indexing has been function inversion. Function inversion is a widely studied data structures problem in cryptography \cite{fiatNaor, GolovnevGuPe23, CorriganKo19, Hellman80}. The goal is to give an optimal tradeoff curve between time $T$ and space $S$ for a data structure inverting a function $f: [N] \to [N]$ with pre-processing time $\tilde{O}(N)$. $T$ bounds the number of queries to $f$ as well as the online computation upon a query. Fiat and Naor \cite{fiatNaor} present a randomized data structure offering a general tradeoff curve of $TS^3 = \tilde{O}(N^3)$, which was recently improved \cite{GolovnevGuPe23} to $TS^2 \max(S, T) = \tilde{O}(N^3)$ (our algorithm, however, performs function inversion on the regime in which these two curves are equivalent). Below, we describe the more specific tradeoff provided by Fiat-Naor function inversion for functions with bounded in-degree.
\paragraph{Fiat-Naor Function Inversion \cite{fiatNaor}.}
\begin{lemma}[\cite{fiatNaor}] \label{lm:fiatnaor}
    Given a function $f: [N] \to [N]$, we can construct a data structure, in pre-processing time $\widetilde{O}(N)$, with time/space $\widetilde{O}(T)$ and $\widetilde{O}(S)$ satisfying the tradeoff $TS^2 = N^3 q(f)$, where $q(f)$ denotes the probability that $f(x) = f(y)$ for two uniformly sampled $x, y \in [N]$. Thus if, for all $y \in [N]$, $|f^{-1}(y)| \leq \Delta$, i.e., $\Delta$ is the maximum in-degree of $f$, we can achieve a time/space tradeoff of $TS^2 = N^2 \Delta$. 
\end{lemma}

We now prove the following corollary, which exactly describes the regime in which our algorithm performs function inversion. The proof utilizes a hashing argument of \cite{dataStructuresFor3sum2020} as well as the more specific Fiat-Naor tradeoff for bounded in-degree functions.
\begin{cor}\label{cor:fiatnaor}
    We can perform Fiat-Naor function inversion for a function $f:U\times V \to W$ satisfying $|W|,|U|\cdot |V| = \widetilde{O}(n^{2-\delta})$, of $\widetilde{O}(1)$ maximum in-degree, in space $S = n^{2-\eta}$ and query time $T = n^{2\eta - 2\delta}$, for any $\eta \in [\delta, 2]$. The pre-processing time is $\widetilde{O}(n^{2-\delta})$.
\end{cor}
\begin{proof}
    We observe that the domain and co-domain are both of size $N = n^{2 - \delta}$, and that $\Delta = \widetilde{O}(1)$. Thus we can hash the domain and co-domain to $[N]$ via the technique of \cite{dataStructuresFor3sum2020}. We briefly summarize the hashing argument here (see Section 3 of \cite{dataStructuresFor3sum2020} for a full exposition). In short, if we sample a uniformly random prime $p_i$ from the interval $[n^2, 6\alpha n^2 \log n]$, we can argue that there is a constant probability for any fixed element $z_1 \in A + B$ in the sumset, that $z_1$ does not collide under $p_i$ with any distinct $z_2 \in A + B$, where we say $z_1$ and $z_2$ collide under $p_i$ if $z_1 \equiv_{p_i} z_2 $. This follows via a counting argument, as there are at least $2\alpha n^2$ primes in the range (for sufficiently large $n$, by the prime number theorem), and at most $\binom{n}{2} \log_{n^2} n^\alpha < n^2 \alpha/4$ of them divide any $z_2 - z_1$ over all $z_2 \in A + B$. Thus, a uniformly random prime $p$ from the aforementioned interval will avoid any collisions between $z_1 \in A+B$ and all other $z_2 \in A + B$ with constant probability. Thus we can sample $\log n$ such primes during pre-processing, and check that every $c \in A + B$ avoids collisions under some $p_i$. This constraint will guarantee no false negatives are reported, and we avoid false positives by checking $3$-sum solutions against the actual original $(a, b, c)$ triplet before reporting. We then maintain a Fiat-Naor data structure$\pmod {p_i}$ for each prime $p_i$, at the cost of a $\log n$ time/space blow up. (See the proof of Theorem 3.3 in \cite{dataStructuresFor3sum2020} for a detailed accounting that $\log n$ primes is sufficient, and see the proof of Lemma 3.2 in \cite{dataStructuresFor3sum2020} for a detailed description of Fiat-Naor function inversion for \threesum{} modulus $p$.) By applying Lemma \ref{lm:fiatnaor}, we can then achieve time/space tradeoff $TS^2 = \tilde{O}((n^{2 - \delta})^2)$ for function inversion, with pre-processing time $\tilde{O}(N) = \tilde{O}(n^{2-\delta})$. The choices $S = n^{2 -\eta}$ and $T = n^{2\eta - 2\delta}$ satisfy the equation (where $\eta$ is a free variable).
\end{proof}

    \section{Preprocessed 3SUM for Unknown $C$}
\label{sec:main-result}
We can now prove our main result. As noted in the technical overview, our algorithm handles values differently based on their multiplicity. We describe each part of the algorithm separately.
\SubquadraticUnknownC*

In fact, our algorithm returns for every $c\in C'$ "yes" if there exists a pair $a,b\in A'\times B'$ that sums to $c$ and "no" otherwise.

\subsection{Heavy Hitters}
Let $\delta>0, r>1$ be parameters to be set later. We begin by handling the values of $A+B$ with multiplicity greater than $n^\delta$. Define the set of \emph{heavy hitters} as follows.

\begin{definition}[Heavy hitters]
    $\widetilde{C}\coloneqq\{c\in A+B : |\{(a,b)\in A\times B : a+b = c\}|\geq n^\delta\}.$
\end{definition}

As every element in $\widetilde{C}$ has $\geq n^\delta$ pairs in $A\times B$ that sum to it, we can bound the size of the set by $|\widetilde{C}|\leq n^{2-\delta}$.
We repeat the following algorithm in parallel $ \log n$ times in both the preprocessing and query phases.

\begin{mdframed}
\begin{description}
\item[Preprocessing:] To preprocess the heavy hitters, our algorithm first computes and stores the set $\widetilde{C}$. Next, we randomly sample a prime $p$ from the range $[n^r, 2n^r)$, compute and store the set of false positives in $A\times B\times \widetilde{C}$ mod $p$:
\[
F \coloneqq \{(a,b,c)\in A\times B \times \widetilde{C} : a+b\equiv_p c \land a+b\neq c\}.
\]

We store the set $F$ as a dictionary indexed by each tuple's $c$ coordinate, so that given a value $c$ we can find all triples $(a,b,c)\in F$ in constant time per triple.

\end{description}
\end{mdframed}

We can construct $F$ in time $\widetilde{O}(n^2 + |F|)$ and using \cref{lm:falsepositives} with $R=n^r$, bound the size of $F$ by $\widetilde{O}(\frac{n^2\cdot |\widetilde{C}|}{p}) \leq  \widetilde{O}(n^{4-\delta - r})$. 
In total, this preprocessing step stores the sets $\widetilde{C}$ and $F$, using space $\widetilde{O}(n^{2-\delta} + n^{4-\delta - r})$.

\begin{mdframed}
\begin{description}

\item[Query:] Given the sets $A',B',C'$, compute the multiset $H \coloneqq A'+B'$ mod $p$, using \cref{lm:fft} in time $\widetilde{O}(p)=\widetilde{O}(n^r)$. For every $c\in C'\cap \widetilde{C}$, query all triples $(a,b,c)\in F$. Count the number of triples such that $a\in A',b\in B'$. If this number is smaller than $H[c]$, return "yes" for $c$, otherwise return "no".
\end{description}
\end{mdframed}

We note that at query time, our algorithm does not read the entire set $F$, but only the false positives $F'\subseteq F$ such that their $c$ value is in $C'$, 
\[
F' \coloneqq \{(a,b,c)\in A\times B \times (\widetilde{C} \cap C') : a+b\equiv_p c \land a+b\neq c\}.
\]

For any $c\in \widetilde{C}\cap C'$, if there exists a pair $(a,b)\in A'\times B'$ that sums to $c$, then $H[c]$ will be greater than the number of triples $(a, b, c)$ in $F$ that have $c$ as their third coordinate with $(a, b) \in A' \times B'$. This is because $H[c]$ counts all triples $(a,b,c)$ with $(a, b) \in A' \times B'$ where $a+b \equiv_p c$, which includes any triples in $F$ with first two coordinates from $A' \times B'$ and third coordinate $c$ as well as at least one true (non-false-positive) sum $(a,b,c)$. Hence, the algorithm will output "yes" on any such value $c\in \widetilde{C}\cap C'$.

By \cref{lm:falsepositives}, we can bound the size of $F'$ by $\widetilde{O}(\frac{n^2 \cdot |\widetilde{C}\cap C'|}{p})\leq \widetilde{O}(n^{3-r})$. By Markov's inequality, our algorithm can run in $\widetilde{O}(n^{3-r})$ time and scan all of $F'$ with constant probability and so repeating this $\log n$ times will guarantee success with high probability. Thus, the total runtime of this query algorithm is $\widetilde{O}(n^r + n^{3-r})$.

\subsection{Non-Heavy Hitters}
Next we handle the non-heavy hitters. Let $\eta > 0$ be a parameter to be set later. Define $\ell = \lceil n^\delta \log n \rceil$. 
We repeat the following algorithm in parallel $4 \log n$ times in both the preprocessing and query phases.

\begin{mdframed}
\begin{description}
\item[Preprocessing:] Construct the sets $A_1, A_2, \ldots, A_\ell$ by assigning each element of $A$ to one of the sets $A_i$ uniformly at random.
We would like to preprocess a Fiat-Naor function inversion instance on each tuple $A_i\times B$. However, we would like to claim that $f_i:A_i \times B\to A_i+B$ defined by $(a,b)\mapsto a+b$ has bounded in-degree. To obtain this, compute $M = \max (A+B)$ and define:
\[
f(a,b) = 
\begin{cases}
    a+b & a+b \notin \widetilde{C}, \\
    a\cdot M + b & a+b \in \widetilde{C}.
\end{cases}
\]
To be able to compute $f$, we store $M$ and $\tilde{C}$ using expected $\widetilde{O}(n^{2-\delta})$ space. As we show in \cref{claim:f-bound}, with probability at least $1 - \frac{\ell}{n^{6}}$, each $f_i$ has in-degree at most $O( \log n)$.

Now, for every $i\in [\ell]$, preprocess a bounded-in-degree Fiat-Naor function inversion data structure $Q_i$ on $f_i \coloneqq f|_{A_i\times B}$ using \cref{cor:fiatnaor}. This uses $\widetilde{O}(n^{2\eta-2\delta})$ query time and $\widetilde{O}(n^{2-\eta})$ space.

\end{description}
\end{mdframed}

Note that since all values of $A,B$ are positive, $f$ is an injective function on pairs $(a,b)$ such that $a+b\in \widetilde{C}$.
In total, we store $\widetilde{O}(n^\delta)$ instances of a Fiat-Naor data structure, as well as $|\widetilde{C}|$ values of $U$. Thus, this preprocessing step uses expected $\widetilde{O}(n^{2-\delta} + n^{2 - \eta + \delta})$ space.

\begin{mdframed}
\begin{description}
\item[Query:] Given the sets $A',B',C'$, perform the following query. For every $c\in C'\setminus \widetilde{C}$, query $f_i^{-1}(c)$ for every $i\in [\ell]$ using the preprocessed function inversion data structure $Q_i$. If a pair $(a,b)\in A_i\times B$ is returned, check if $(a,b)\in A'\times B'$ and if so, return "yes" for this $c$. If no such pair is found for any $i$, return "no".

\end{description}
\end{mdframed}

In total, for every $c\in C'\setminus \widetilde{C}$ we perform $\widetilde{O}(n^\delta)$ calls to query a Fiat-Naor function inversion data structure for a total query time of $\widetilde{O}(n\cdot n^\delta \cdot n^{2\eta - 2\delta}) = \widetilde{O}(n^{1+2\eta - \delta})$.

The following two claims show the correctness of this stage of the algorithm. First we show that the functions $f_i$ do in fact have low in-degree, justifying our use of \cref{cor:fiatnaor}.

\begin{claim}\label{claim:f-bound}
    In each repetition, the $\ell$ functions $f_i$ have in-degree at most $11 \ln n$ with probability at least $1 - \frac{\ell}{n^{6}}.$
\end{claim}
\begin{proof}
    It suffices to argue that for every $c \notin \widetilde{C}$, the in-degree of $f_i$ for $c$ is at most $11 \ln n$ with high probability. 

    Fix one $c \notin \widetilde{C}$ and one $f_i$. Denote by $\mu$ the expected number of pre-images $X$ for $f_i$ for $c$.
 By applying a Chernoff bound, we have that
    
   \[
   \Pr(X\geq 11\ln n)\leq \Pr\left(X\geq \left(1+\frac{10\ln n}{\mu}\right)\mu \right) \leq e^{-\frac{100 \ln^2 n /\mu}{10\ln n /\mu + 2}}.
   \]
   Bounding further, we obtain
 $
  \Pr(X\geq 11\ln n)\leq e^{-\frac{100 \ln^2 n}{10\ln n + 2\mu}}$.
   For each $c \notin \widetilde{C}$, note that the expected number $\mu$ of pre-images $X$ for $f_i$ for $c$ is less than 1. Hence:
  
\[\Pr(X\geq 11\ln n)
  \leq e^{-\frac{100 \ln^2 n}{10\ln n + 2}}\leq e^{-\frac{100 \ln^2 n}{12\ln n}}\leq 1/n^8.\]

    Therefore, taking a union bound over all (at most $n$) $c \notin \widetilde{C}$ and all (at most $\ell$) functions $f_i$, we have that with probability at least $1 - \frac{\ell}{n^{6}}$, the functions $f_i$ all have in-degree at most $11 \ln n = O(\log n)$. This allows us to use \cref{cor:fiatnaor} and preprocess a Fiat-Naor function inversion data structure for functions with logarithmically bounded in-degree.
\end{proof}

Next, we show that this process finds every pair that sums to a value $c\in C'\setminus \widetilde{C}$. Thus, if such a pair exists in $A'\times B'$ the algorithm will be able to find it.

\begin{claim}
    For every $(a, b, c)$ such that $a + b = c$ and $c \notin \widetilde{C}$, a query to $c$ uncovers $(a, b)$ with probability at least $1 - \frac{1}{n^4}.$
\end{claim}
\begin{proof}
    First, we consider a single repetition of the algorithm.
    
    Suppose $(a, b) \in A_i \times B$. There are at most $n^\delta$ other pairs $(a', b')$ such that $a + b = a' + b'$. For each such $a', b'$, the probability that $a' \in A_i$ is $\frac{1}{n^\delta \log n}$. Since there are at most $n^\delta$ such $a'$, the probability that all of them are not in $A_i$ is at least $(1 - \frac{1}{n^\delta \log n})^{n^\delta} \geq 1 - \frac{1}{\log n}$ by Bernoulli's inequality. 

    Since we repeat this experiment $4 \log n$ times, the probability that at least one of the repetitions outputs $(a, b)$ is
    \[
        1 - \left(\frac{1}{\log n}\right)^{4 \log n} \geq 1 - \frac{1}{n^4}.
    \]
\end{proof}

When given a query $A',B',C'$, if $a\in A',b\in B',c\in C'$ then the algorithm will check $(a,b)$ and thus return "yes" with probability $\geq 1-\frac{1}{n^4}$. By a union bound, if any such triple exists, $(a,b,c)\in A'\times B'\times C'$ such that $a+b=c$ and $c\notin \widetilde{C}$, the algorithm will output "yes" with high probability.

\subsection{Putting it All Together}
Combining the two algorithms above, given two sets $A,B$ we preprocess them to obtain the data structure for the heavy hitters and non-heavy hitters cases. This uses a total space of $\widetilde{O}(n^{2-\delta} + n^{4-\delta - r} + n^{2-\eta + \delta})$.

Upon query, given $A',B',C$, we perform the two query algorithms described above to handle $C'\cap \widetilde{C}$ and $C'\setminus \widetilde{C}$ separately, for a total query time of $\widetilde{O}(n^r + n^{3-r} + n^{1+2\eta - \delta})$.

Setting $\eta = \frac{1}{2} + \frac{\eps}{3}, r = \frac{3}{2} + \eps, \delta = \frac{1}{2}-\frac{\eps}{3}$ we get a query time of $\widetilde{O}(n^{1.5 + \eps})$ using space $\widetilde{O}(n^{2-2\eps/3})$.

The preprocessing time is dominated by constructing $\widetilde{C}$, in $\widetilde{O}(n^2)$ time, constructing $F$ in $\widetilde{O}(n^2 + |F|) = \widetilde{O}(n^2)$ time, and constructing $\widetilde{O}(n^\delta)$ instances of a Fiat-Naor function inversion data structure each in time $\widetilde{O}(n^{2-\delta})$. Thus, in total, our algorithm takes $\widetilde{O}(n^2)$ preprocessing time.
    \section{Discussion and Future Work}
We presented the first data structure for preprocessed \threesum{} with unknown $C$ that has truly subquadratic space and query time (and quadratic preprocessing).
We'll discuss some possible extensions.

\paragraph{What do we get with optimal function inversion?} The only known lower bounds for function inversion \cite{Yao90,DeTT10,GennaroGKT05} give that the space $S$ and query time $T$ must satisfy $$ST\geq \tilde{\Omega}(N).$$ While a matching upper bound is only known for random permutations \cite{Hellman80}, it is conceivable that a construction achieving this bound might be possible for arbitrary functions as no higher lower bounds are known.

Using a hypothetical optimal function inversion data structure achieving $ST=N$, just by plugging into our construction, we would obtain for every $\eps\in [0,1/2]$, a data structure for preprocessed \threesum{} with unknown $C$ with 
\begin{center}space $n^{2-\eps}$ and query time $n^{1.5+\eps}$.\end{center}

\paragraph{What tradeoff is desirable?}
A brute-force query algorithm would solve \threesum{} on each query without the need for preprocessing. This would achieve space $O(n)$ (just store $A$ and $B$) and query time $O(n^2)$\footnote{One can also shave off polylog factors as in \cite{bdp}.}. Interpolating between this and the space-$n^{2}$ and query-$n^{1.5}$ trade-off of \cite{preprocessedSOSA2025}, one might hope to be able to achieve for every $\eps\in [0,1/2]$,
\begin{center}space $n^{2-2\eps}$ and query time $n^{1.5+\eps}$.\end{center}

Unfortunately, even with optimal function inversion, our construction doesn't achieve such a result. 
Either new ideas are needed, or perhaps there is some limitation to achieving such a tradeoff.

\paragraph{Possible extensions?}
Our approach can be useful for other potential problems with preprocessing, beyond \threesum{}.
We give \threexor{} as an example. 

In the \threexor{} problem, one is given three $n$-sized sets of {\em binary strings} $A,B,C\subseteq \{0,1\}^d$. We want to know whether there exist $a\in A, b\in B, c\in C$ such that $a\oplus b\oplus c=0$. \threexor{} can be solved in $O(n^2)$ time deterministically and, with randomness, in expected time $O(n^2 (\log\log n)^2 / \log^2 n)$ \cite{DietzfelbingerScWa18}. It is conjectured (e.g., \cite{jaf}) that $n^{2-o(1)}$ time is necessary.

Analogous to \threesum{}, we can define the preprocessed \threexor{} with unknown $C$ problem as: Given $A,B\subseteq \{0,1\}^d$, preprocess them so that the following query can be supported efficiently: given $A'\subseteq A, B'\subseteq B$ and $C\subseteq \{0,1\}^d$, decide if there exist $a\in A', b\in B', c\in C$ such that $a\oplus b\oplus c=0$.

We can obtain essentially the same trade-off for \threexor{} as for \threesum{} for this preprocessed version, just by simulating some key steps of the \threesum{} algorithm with \threexor{} primitives.

The first primitive we need to simulate is instance reduction with a small number of false positives. For \threesum{} this was, given some $p$, to efficiently compute $A+B\mod p$ (with multiplicities) and also to pick $p$ so that the number of false positives is small. 

While XORing vectors mod a prime is not defined, Jafargholi and Viola \cite{jaf} provided an analogous tool for \threexor{}. They defined linear hash functions that take a vector $a\in \{0,1\}^d$ to a vector $h(a)\in \{0,1\}^{\ell}$ for some choice $\ell<<d$. This is essentially achieved by random linear combinations of the vector components. Because of linearity, if $a\oplus b\oplus c=0$, we also have $h(a)\oplus h(b)\oplus h(c) = 0$. Jafargholi and Viola \cite{jaf} showed that the expected number of false positives is $O(n^3/2^\ell)$. Thus, effectively, we can think of $p:=2^\ell$.

For \threesum{} we use the FFT to compute $A+B\mod p$ (with multiplicities) in $\tilde{O}(p)$ time. For \threexor{} one can use the Walsh-Hadamard transform to 
 compute for all $c\in \{0,1\}^{\log p}$ the number of pairs $(a,b)\in A\times B$ such that $c=a\oplus b$, i.e., the number of \threexor{} solutions that $c$ is part of. The running time is $O(p\log p)$  using the fast folklore recursive algorithm to compute the transform on vectors of length $\log p$.

The rest of our algorithm (Fiat-Naor and Walsh-Hadamard instead of FFT) can be simulated for \threexor{} exactly as for \threesum{}, and hence we obtain a trade-off of space usage $\tilde{O}(n^{2-2\eps/3})$ and query time $\tilde{O}(n^{1.5+\eps})$ after quadratic time preprocessing, for preprocessed \threexor{} with unknown $C$.

The approach can be used for other problems, provided the problems admit analogues of instance size reduction (like mod $p$) with bounded false positives and of the FFT for computing for every $c$ in the range, the number of solutions in the reduced instance that use $c$.

\section*{Acknowledgments}
This research was supported by NSF Grant CCF-2330048, BSF Grant 2024233, and a Simons Investigator Award.
Y. Kirkpatrick was supported in part by an NSF Graduate Research Fellowship under Grant No 2141064. J. Kuszmaul was supported in part by an MIT Akamai Presidential Fellowship. S. Mathialagan was supported in part by NSF CNS-2154149, a Simons Investigator Award, and Jane Street.

    \bibliographystyle{alphaurl} 
    \bibliography{refs}
\end{document}